\documentclass[10pt,orivec]{llncs}

\usepackage{amssymb, stmaryrd}
\usepackage{graphicx, amsmath}
\usepackage{subfigure}
\usepackage{color}
\usepackage{tikz}
\usepackage{verbatim}
\usepackage{cite}
\usepackage{graphicx}
\usepackage{algorithmic}
\usepackage{algorithm}
\usepackage{diagbox} 
\usepackage{subfigure}

\renewcommand{\epsilon}{\varepsilon}
\newcommand{\functiondot}{\,\mathord{\cdot}\,}
\newcommand{\bigO}[1]{\mathcal{O}(#1)}

\newcommand{\setnocond}[1]{\{#1\}}
\newcommand{\setcond}[2]{\{\,#1 : #2\,\}}

\newcommand{\norm}[1]{\|#1\|}
\newcommand{\opnorm}[1]{\norm{#1}_{o}}

\newcommand{\radius}[1]{\varrho(#1)}
\newcommand{\maxmag}[1]{\kappa(#1)}

\newcommand{\epsneigh}[1][\epsilon]{U_{#1}}
\DeclareMathOperator{\modulo}{mod}
\DeclareMathOperator{\lcm}{lcm}

\newcommand{\naturals}{\mathbb{N}}
\newcommand{\reals}{\mathbb{R}}

\renewcommand{\dh}{\ensuremath{\mathcal{D(H)}}}
\newcommand{\sh}{\ensuremath{\mathcal{S(H)}}}
\newcommand{\bh}{\ensuremath{\mathcal{B(H)}}}
\newcommand{\ap}{AP}

\renewcommand{\b}{\ensuremath{\mathcal{B}}}
\renewcommand{\c}{\ensuremath{\mathcal{C}}}
\newcommand{\e}{\ensuremath{\mathcal{E}}}
\newcommand{\f}{\ensuremath{\mathcal{F}}}
\newcommand{\g}{\ensuremath{\mathcal{G}}}
\newcommand{\h}{\ensuremath{\mathcal{H}}}
\renewcommand{\i}{\ensuremath{\mathcal{I}}}
\renewcommand{\l}{\ensuremath{\mathcal{L}}}
\newcommand{\p}{\ensuremath{\mathcal{P}}}
 
\newcommand{\s}{S}
\newcommand{\bbh}{\ensuremath{\mathcal{B(\bh)}}}

 \newcommand{\trj}{\sigma}

\newcommand{\vp}{\varphi}
\newcommand{\coloneqq}{:=}
\newcommand{\Coloneqq}{::=}
 \newcommand{\tr}{{\rm tr}} 

\newcommand{\spec}{\mathit{spec}}

\newcommand{\bra}[1]{\langle #1 |}
\newcommand{\ket}[1]{| #1 \rangle}
\newcommand{\braket}[2]{\langle #1 | #2 \rangle}

\def \trjst {\sigma_s}
\def \trjdy {\sigma_d}
\def \lfst {L_s}
\def \lfdy {L_d}
\def \mdst {\models_s}
\def \mddy {\models_d}

\title{Model Checking Applied to Quantum Physics}

\titlerunning{Model Checking Applied to Quantum Physics}
\author{Ji Guan\inst{1,2}, Yuan Feng\inst{2}, Andrea Turrini\inst{1}\and  Mingsheng Ying\inst{1,2,3} }
\authorrunning{J. Guan et al.}
\institute{
State Key Laboratory of Computer Science, Institute of Software, Chinese Academy of Sciences, Beijing 100190, China\and Center for Quantum Software and Information, University of Technology Sydney, NSW 2007, Australia \and Department of Computer Science and Technology, Tsinghua University, Beijing 100084, China}
\begin{document}
\maketitle
\pagestyle{plain}
\begin{abstract}
 Model checking has been successfully applied to verification of computer hardware and software, communication systems and even biological systems.  
In this paper, we further push the boundary of its applications and show that it can be adapted for applications in quantum physics. More explicitly, we show how quantum statistical and many-body systems can be modeled as quantum Markov chains, and some of their properties that interest physicists can be specified in  linear-time temporal logics. Then we present an efficient algorithm to check these properties. A few case studies are given to demonstrate the use of our algorithm to actual quantum physical problems.   
\end{abstract}

\section{Introduction}

Classical mechanics describes nature at macroscopic scale (far larger than $10^{-9}$ meters), while  quantum mechanics is applied at microscopic scale (near or less than $10^{-9}$ meters). A particle at this level can be mathematically represented by  a normalized complex vector $\ket{\psi}$ in a Hilbert space $\h$. The time evolution of a single particle system is described by the Schr\"{o}dinger equation: 
\begin{equation}
	i\frac{d\ket{\psi_t}}{dt}=H\ket{\psi_t}
\end{equation}
with some \emph{Hamiltonian} $H$ (a Hermitian matrix on $\h$), where $\ket{\psi_t}$ is the state of the system at time $t$. In practice, suffering from noises, the state of a quantum system cannot be completely known. Thus a \emph{density operator} $\rho$ (Hermitian  positive semidefinite matrix with unit trace) on $\h$ is introduced to describe the uncertainty of the possible states: \[\rho=\sum_k p_k\ket{\psi_k}\bra{\psi_k},\]
where $\{(p_k,\ket{\psi_k})\}_k$ is a mixed state or an ensemble expressing that the quantum state is at $\ket{\psi_k}$ with probability $p_k$, and $\bra{\psi_k}$ is the conjugate  transpose of $\ket{\psi_k}$. In this case, the evolution is described by the Lindblad equation:
\begin{equation}\label{Lind}
	\frac{d\rho_t} {dt}=L(\rho_t)
\end{equation}
where $\rho_t$ stands for the (mixed) state of the system, and $L$ is a  linear function of $\rho_t$, which is generally irreversible.

\subsection{Two Model Checking Problems from Quantum Physics}\label{2pro}

Our motivations are two problems from two different fields of quantum physics: 

{\vskip 4pt}

\textbf{Quantum Statistical Mechanics}: 
\emph{Quantum statistical mechanics} is essentially statistical mechanics applied to quantum systems. It is based on the statistical description of measurements~\cite{von2018mathematical}. 
Specifically, through observing state $\rho_t$ of a quantum system at time $t$ with a \emph{quantum measurement} (e.g., position and momentum), which is mathematically modeled by a set $\{M_{k}\}_{k=1}^m$ of matrices on its state Hilbert space $\h$ with $\sum_k M_k^\dagger M_k=I$ (the identity operator on $\h$), the probability of outcome $k$ is 
\[p_k=\tr(M_k^\dagger M_k\rho_t).\]
After observing $k$, the state becomes
\[\rho'_t(k)=\frac{M_k\rho_t M_k^\dagger}{\tr(M_k^\dagger M_k\rho_t)}.\]
 The vital difference to classical statistical mechanics is that the original state $\rho_t$ is collapsed (changed) to $\rho'_t(k)$ after we measure the system, depending on the measurement outcome $k$.
 
Quantum statistical mechanics is mainly concerned with the connections between the classical information (probability distributions of measurement outcomes) and the quantum information (quantum states) of quantum systems. A fundamental problem in the foundations of quantum statistical mechanics is the long-term classical information of quantum systems. It originated from John von Neumann's 1929 paper on the quantum ergodic theorem \cite{neumann1929beweis,goldstein2010long}, which asserts that for an appropriate finite set of mutually commuting measurements, every initial quantum state evolves so that for most time in the long run, the joint probability distribution of these measurements is close to a certain distribution.
A renewed interest in recent years leads to the study of long-term properties of the measurement outcome distribution $(\tr(M_1^\dagger M_1\rho_t),\tr(M_2^\dagger M_2\rho_t),\ldots,\tr(M_n^\dagger M_n\rho_t))$ \cite{rigol2008thermalization,popescu2006entanglement,linden2009quantum}; especially,
 \begin{problem}[\textbf{Long-term classical information}]\label{pro:statistic}
   Let $\{\i_l\}_{l=1}^n$ be a finite set of  intervals in $[0,1]$ and $\{M_k\}_{k=1}^m$ a measurement. Given a multiset $\{l_k: 1\leq l_k \leq n\}_{k=1}^m$, is $\tr(M_k^\dagger M_k\rho_t)$  eventually (respectively, infinitely often) in $\i_{l_k}$ for all $1\leq k\leq m$?
 \end{problem}  
 
\textbf{Quantum Many-body Systems}: A \emph{quantum many-body system} is a complex system of multiple interacting microscopic particles \cite{zeng2015quantum}. The number of particles can be near or more than $10^{20}$ when we consider \emph{thermodynamic limit} (of \emph{quantum condensed matter}) in practice, and the dimension of the state space for the whole system (all particles) is at least $2^{10^{20}}$.
\emph{Quantum many-body problems} are concerned with bulk properties (e.g., superfluidity and superconductivity) of such large systems. 
Obviously, exact or analytical solutions to them are impractical or even impossible. A common approach is to find hypothetical models that capture some essential aspects (e.g., ground states and ground energy) of the real systems, such as \emph{Matrix Product States (MPS)}  and \emph{Tensor Product States (TPS)} in terms of the topological structure of the systems (see Fig.1) \cite{zeng2015quantum}. 
\begin{figure}[htbp]\label{pic}
\centering
\subfigure[]{
\begin{minipage}[t]{0.5\linewidth}
\centering
\includegraphics[width=1.9in,height=1.5in]{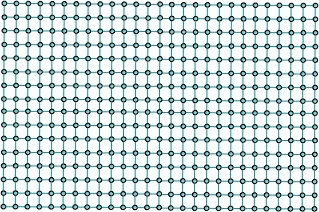}
\end{minipage}%
}%
\subfigure[]{
\begin{minipage}[t]{0.5\linewidth}
\centering
\includegraphics[width=1.9in,height=1.5in]{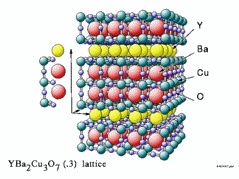}
\end{minipage}%
}%
\centering
\caption{(a) and (b) are 2- and 3-dimensional quantum many-body systems (lattice), respectively. Each node represents a particle.}
\end{figure}
Let us consider an $1$-dimensional MPS as an example. Assume the system consists of $N$ quantum particles in a line, indexed from $1$ to $N$, and 
each particle has its own $d$-dimensional Hilbert space, denoted by $\h_{d}$. 
Then the entire Hilbert space is $\h = \h_{d}^{\otimes N}$. MPSs have the form:
\begin{eqnarray}\label{model-MPS}
    \ket{\psi_{N}} = \sum_{k_{1},\dotsc,k_{N}=1}^{d} \tr(A_{k_{1}} \cdots A_{k_{N}}) \ket{k_{1}}\otimes \cdots \otimes\ket{k_{N}}
    \label{eq:mps}
\end{eqnarray}
where $\setnocond{\ket{k}}_{k = 1}^{d}$ is an orthonormal basis of $\h_{d}$ and $\setnocond{A_{k}}_{k=1}^{d}$ is a family of $D\times D$ complex matrices with $D$ being independent on $N$. 

For each non-zero $\ket{\psi_{N}} \in \h$, there always exists a \emph{parent Hamiltonian} $H_{N}$ which has $\ket{\psi_{N}}$ as a \emph{ground state} -- the eigenvector of $H_{N}$ corresponding to the smallest eigenvalue. Such an $H_{N}$ can be constructed from a locally translationally invariant Hamiltonian $h$ as $H_{N} = \sum_{k=1}^{N} \tau^{k}(h)$, where $\tau$ is the one-step translation operator and $h$ is the projector onto the orthogonal complement of 
\[
    \left\{\sum_{k_{1}, \dotsc, k_{J}} \tr(XA_{k_{1}} \cdots A_{k_J}) \ket{k_{1}} \otimes \cdots\otimes\ket{k_{J}}\right\}
\]
for some sufficiently large interaction range $J$ -- a positive integer, where $X$ ranges over $D\times D$ complex matrices; for more details, we refer to~\cite{Perez-Garcia2007}.  

When considering MPSs in thermodynamic limit ($N\rightarrow \infty $), we are concerned with a family of parent Hamiltonians $\setnocond{H_{N}}_{N \geq J}$ such that $H_{N} \ket{\psi_{N}} = 0$. 
However, $\ket{\psi_{N}}$ is not a ground state if $\ket{\psi_{N}} = 0$. Therefore, for verifying the validity of matrices $\{A_k\}$ in the hypothetical model (\ref{model-MPS}) of an MPS and its parent Hamiltonian, we need to answer:
\begin{problem}[\textbf{Dichotomy problem}]\label{pro:many_body}
	Given a finite set of $D\times D$ complex matrices $\setnocond{A_{k}}$ and a positive integer $J$. 
	\begin{itemize}
		\item Is $\ket{\psi_{N}} \neq 0$ for all $N \geq J$?
		\item Is there $N_{0} \geq J$ such that $\ket{\psi_{N}} \neq 0$ for all $N \geq N_{0}$?
	\end{itemize}
\end{problem}

\subsection{Contributions of the Paper}
In this paper, we show that quantum statistical and many-body systems can be modeled as \emph{quantum Markov chains (QMCs)} if we are interested in their discrete-time evolutions. Furthermore, we show that the properties considered in  Problems~\ref{pro:statistic} and~\ref{pro:many_body} can be properly specified in a \emph{linear-time temporal logic (LTL)}, and thus these two problems are typical \emph{LTL model checking} problems for QMCs. 

QMCs have been introduced as quantum generalizations of classical Markov chains in several different areas, including quantum control \cite{Ticozzi2008a}, quantum information theory \cite{guan2018structure}, and quantum programming \cite{ying2016foundations}.   
A QMC is defined as a $3$-tuple $\g = (\h, \e, \rho_{0})$, where $\h$ is a finite-dimensional Hilbert space, $\rho_0$ is the initial state, and $\e$ is a \emph{completely positive and trace-preserving} (CPTP) map  (also called a \emph{super-operator}) on $\h$. Intuitively, $\e$ models the system's dynamics and transforms a state (density operator) $\rho$ to another state $\e(\rho)$. It can be understood as a discrete-time solution to the Lindblad equation (\ref{Lind}).   
Moreover, several model checking-related problems for QMCs have been studied, including the long-run behavior and reachability problem; for example, some characterizations of the limiting states of QMCs were given in \cite{wolf2012quantum,guan2018decomposition}; several algorithms for computing the probabilities of reachability, repeated reachability and persistence of QMCs were presented in~\cite{ying2013reachability,Baumgartner2011,guan2018decomposition} based on irreducible and periodic decomposition techniques. However, these results cannot be used to solve Problems~\ref{pro:statistic} and~\ref{pro:many_body}.   

The aim of this paper is to develop an efficient model checking algorithm solving Problems~\ref{pro:statistic} and~\ref{pro:many_body}. 
Our basic idea is inspired by Thiagarajan's \textit{approximate verification} of classical \emph{Markov chains}~\cite{agrawal2015approximate} where concrete \emph{atomic propositions} are used to estimate the actual distribution. 
For the flexibility of applications, we admit \textit{abstract} LTL formulas, which can specify the properties in Problems~\ref{pro:statistic}  and~\ref{pro:many_body}. 
We further give an effective procedure to approximately answer the LTL model checking problem for \emph{periodically stable} QMCs.  
The main technique is based on the \emph{eigenvalue-analysis} of QMCs, which significantly simplifies the previous work based on decompositions of the state space. 

Several case studies are provided in Section \ref{sec:experiments} to illustrate how our model checking algorithm can be applied in quantum  physics.

\section{Quantum Markov Chains}\label{sec:preliminaries}

For the convenience of the reader, we review some basics of QMCs. For more details, we refer the interested reader to~\cite{nielsen2010quantum,guan2018decomposition}.

\subsection{Dynamics of Quantum Systems} 

The state space of a quantum system is a Hilbert space $\h$. In this paper, we always assume that $\h$ is finite-dimensional. Let $\bh$ be the set of linear operators (matrices) on $\h$. 
A density operator is a positive semi-definite operator $\rho\in \bh$  with $\tr(\rho) = 1$, where $\tr(\rho)$ is the trace of $\rho$, i.e., the summation of diagonal elements of $\rho$. 
A \emph{super-operator} $\e$ on $\h$ is a linear operator on $\bh$. It is called \emph{trace-preserving} if $\tr(\e(\rho)) = \tr(\rho)$ for all $\rho$;
it is \emph{completely positive} if for any Hilbert space $\h'$, the trivially extended operator $\mathrm{id}_{\h'} \otimes \e$ maps density operators in $\b(\h' \otimes \h)$ to density operators, where $\otimes$ denotes the tensor product and $\mathrm{id}_{\h'}$ is the identity map on $\b(\h')$. In this paper, we assume all super-operators to be completely positive and trace-preserving (CPTP).

Let $\dh$ and $\sh$ be the sets of density operators and super-operators on $\h$, respectively. 
According to the postulates of quantum mechanics, $\dh$ represents all valid states (density operators) of the system, and $\sh$ models all the possible (discrete-time) dynamics of the system.
 By Kraus representation theorem~\cite{choi1975completely}, for any super-operator $\e$ on $\h$, there exist linear operators $\setcond{E_{k}}{1 \leq k \leq N}$ with $N \leq \dim(\h)^{2}$ and $\sum_{k} E_{k}^{\dag} E_{k} = I$, such that
\[
    \e(A) = \sum_{k=1}^{N} E_{k} A E_{k}^{\dag},
\]
for all $A \in \bh$, where $\dag$ denotes the Hermitian adjoint. 
In this paper, we sometimes use the Kraus operators to represent a super-operator.
 
\subsection{Quantum Markov Chains}
\label{ssec:qmc}

Recall that a Markov chain (MC) is a tuple $\c = (\s, P, \mu_{0})$, where $\s $ is a finite state set, $P$ a  probability matrix describing the transition probabilities between states, and $\mu_{0}$ a distribution of the initial states. 
Thus, the execution of $\c$ is a set 
of state paths, each one occurring with a certain probability. 
Alternatively, it can be seen as a single path (of distributions): 
$\trj(\c) = \mu_{0},\mu_{0}P,\mu_{0}P^{2},\dotso$. 

QMCs are a straightforward generalization of MCs. 

\begin{definition}
    A QMC is a tuple $\g = (\h, \e, \rho_{0})$, where $\h$ is a  Hilbert space, $\e$ a super-operator on $\h$, and $\rho_{0} \in \dh$ an initial state. 
\end{definition}

Especially, $\g$ is called \emph{irreducible} if $\e$ has only one full-rank stationary state~\cite{Fagnola2009}; that is, there exists a unique $\rho \in \dh$ such that $\e(\rho) = \rho$, and further $\rho > 0$, i.e., $\rho$ is strictly positive.

The state transitions of $\g$ can be described as the \emph{trajectory}: \begin{equation}\label{q-traj}\trjst(\g) = \rho_{0}, \e(\rho_{0}), \e^{2}(\rho_{0}), \ldots.\end{equation}

Sometimes, quantum states are not our concern in practice. Then a QMC can be defined as a pair $\g=(\h,\e)$ without explicitly specifying the initial state, and its behavior described by the trajectory of super-operators:
 \begin{equation}\label{s-traj}\trjdy(\g) = \mathrm{id}_{\h}, \e, \e^{2}, \ldots.\end{equation}
 
 For more discussions and examples of QMCs, see Appendix \ref{More-QMC}. 

\section{Linear-Time Properties in Quantum Physics}

In this section, we present two \emph{linear-time temporal logics} (LTLs) as languages for specifying properties of quantum physical systems. Our logics are essentially the same as the ordinary LTL except that its atomic propositions are interpreted in quantum physics.  

\subsection{Linear-Time Temporal Logic}
\label{ssec:ltl}

As usual, we assume a finite set $\ap$ of atomic propositions.   
The LTL formulas over $\ap$ are defined by the following syntax (see, e.g.,~\cite{baier2008principles}):
\[
    \vp \Coloneqq  \mathbf{true} \mid a \mid \neg\vp \mid \vp_{1} \lor \vp_{2} \mid \bigcirc\vp \mid \vp_{1} U \vp_{2}
\]
where $a\in \ap$. 
Other standard Boolean operators and temporal modalities like $\lozenge$ (\emph{eventually}) and $\square$ (\emph{always}) can be derived in the usual way.

The semantics of LTL is also defined in a familiar way. For any infinite word $\xi$ over $2^{\ap}$ and for any LTL formula $\vp$ over $\ap$, the satisfaction relation $\xi\models\vp$ is defined by induction on the length of $\vp$:
\begin{itemize}
\item 
    $\xi \models \mathbf{true}$;
\item 
    $\xi \models a$ iff $a \in \xi[0]$;
\item 
    $\xi \models \neg\vp$ iff it is not the case that $\xi \models \vp$ (written $\xi \not\models \vp$);
\item 
    $\xi \models \vp_{1} \lor \vp_{2}$ iff $\xi \models \vp_{1}$ or $\xi \models \vp_{2}$;
\item 
    $\xi \models \bigcirc\vp$ iff $\xi[1+] \models \vp$;
\item 
    $\xi \models \vp_{1} U \vp_{2}$ iff there exists $k \geq 0$ such that $\xi[k+] \models \vp_{2}$ and for each $0 \leq j < k$, $\xi[j+] \models \vp_{1}$.
\end{itemize}
Here   $\xi[k]$ and $\xi[k+]$ denote the $(k+1)$-th element and $(k+1)$-th suffix of $\xi$, respectively. 
The indexes start from zero so that, say, $\xi = \xi[0+]$. Furthermore, the \emph{semantics} of $\vp$ is defined as the language containing all infinite words over $2^{\ap}$ that satisfy $\vp$:
\[
    \l_\omega(\vp) = \setcond{\xi \in (2^{\ap})^{\omega}}{\xi \models \vp}.
\]

\subsection{Atomic Propositions Interpreted in Quantum Statistics}
\label{sec:interpretstatistics}
When using our logic to specify properties of quantum statistical systems, we need to choose appropriate atomic propositions $\ap$ and to properly define the satisfaction relation $\rho\models a$ between quantum states and atomic propositions. As pointed out in Section \ref{2pro}, statistical information about a quantum system comes from a measurement. A \emph{physical observable} is modeled by a Hermitian operator $A$ in
the state Hilbert space $\mathcal{H}$, i.e., $A^\dag =A$. Then a quantum measurement can be constructed
from $A$ as follows. An \emph{eigenvector} of $A$ is a non-zero vector
$|\psi\rangle\in\mathcal{H}$ such that $A|\psi\rangle=\lambda |\psi\rangle$ for
some complex number $\lambda$ (indeed, $\lambda$ must be real when $A$ is
Hermitian). In this case, $\lambda$ is called an \emph{eigenvalue} of $A$. For
each eigenvalue $\lambda$, the set
$\{|\psi\rangle:A|\psi\rangle=\lambda|\psi\rangle\}$ of eigenvectors
corresponding to $\lambda$ together with the zero vector is a subspace of
$\mathcal{H}$. We write $P_\lambda$ for the projection onto this subspace. Then
we have the spectral decomposition \cite[Theorem 2.1]{nielsen2010quantum}: $A=\sum_\lambda\lambda P_\lambda$, where
$\lambda$ ranges over all eigenvalues of $A$. Moreover,
$M=\{P_\lambda\}_\lambda$ is a (projective) measurement. If we perform $M$ on
the quantum system in state $\rho$, then the outcome $\lambda$ is obtained with
probability $p_\lambda=\tr(P_\lambda^\dag
P_\lambda\rho)=\tr(P_\lambda\rho)$, and the expectation of $A$ in
state $\rho$ is $$\llbracket A\rrbracket_\rho=\sum_\lambda
p_\lambda\cdot\lambda=\sum_\lambda\lambda\tr(P_\lambda\rho)=\tr(A\rho).$$
Our atomic propositions are chosen to give an estimation of the expectations of physical observables. 

\begin{definition}\label{def-sat}\begin{enumerate}\item An atomic proposition in a Hilbert space $\h$ is defined as a pair $(A,\i)$, where $A$ is an observable  in $\h$ and  $\i\subseteq \reals$ is an interval. 
\item A state $\rho\in\dh$ satisfies $a:=(A,\i)$, written $\rho\models a$, if the expectation of $A$ in $\rho$ lies in interval $\i$: $\tr(A\rho)=\llbracket A\rrbracket_\rho\in\i$. \end{enumerate}
\end{definition}

Now let us extend the satisfaction relation $\rho\models a$ to $\g\models\vp$ between a QMC $\g$ and a general LTL formula $\vp$.  
To this end, we introduce the labeling function: \begin{equation}\label{l1-fun}\lfst \colon \dh \to 2^{\ap},\qquad \lfst(\rho) = \setcond{a \in \ap}{\rho\models a}\end{equation}
  which  assigns to each quantum state the set of atomic propositions in $\ap$ satisfied by the state. We further extend the labeling function to sequences of quantum states by setting $\lfst(\rho_{0}, \rho_{1} ,\ldots) = \lfst(\rho_{0}), \lfst(\rho_{1}), \ldots$. Then we define: $$\g\mdst\vp\ {\rm if\ and\ only\ if}\ \lfst(\trjst(\g))\in\mathcal{L}_\omega(\vp)$$ where $\trjst(\g)$ is the state trajectory of $\g$ as defined in Eq. (\ref{q-traj}).

\begin{example} Given a quantum measurement $\{M_k\}_{k=1}^m$, we consider a sequence of physical observable $\{A_k=M_k^\dag M_k\}_{k=1}^{m}$ and a finite set of intervals $\{\i_l\}_{l=1}^n$ in $[0,1]$. Let $AP = \{(A_k, \i_l):1\leq k
\leq m, 1\leq l \leq n \}$ with atomic proposition $(A_k, \i_l)$ asserting that expectation $\tr(A_k\rho) \in \i_l$. Then Problem~\ref{pro:statistic} can be rephrased as: \begin{itemize}\item Given a multiple set $\{l_k: 1\leq l_k\leq n\}_{k=1}^m$, is $\g \mdst \bigwedge_{k=1}^m\square (A_k,\i_{l_k})$ (respectively, $\g \mdst \bigwedge_{k=1}^m\square\lozenge (A_k,\i_{l_k})$)?\end{itemize}
\end{example}

\subsection{Atomic Propositions Interpreted in Quantum Many-Body Systems}
\label{Subsec:many-body}

When using our logic to specify properties of quantum many-body systems, atomic propositions $\ap$ need to be chosen in a different way.
First, note that given $\ket{\psi_{N}}$ in Eq.~\eqref{eq:mps}, there exists an orthogonal decomposition $\h_{D} = \bigoplus_{m} \h_{D,m}$ such that $\ket{\psi_{N}}$ can be linearly represented by a set of families of operators $\setnocond{\setnocond{B_{m,j}\in \mathcal{B}(\h_{D,m})}_{j}}_{m}$, where $\e_{m}(\functiondot) = \sum_{j} B_{m,j} \functiondot B_{m,j}^{\dag}$ is a super-operator and $(\h_{D,m}, \e_{m})$ is irreducible, with positive coefficients $\setnocond{a_{m} > 0}_{m}$:
\[
    \ket{\psi_{N}} = \sum_{m} a_{m} \ket{\phi_{N,m}},
\]
where $$\ket{\phi_{N,m}} = \sum_{k_{1}, \dotsc, k_{N}} \tr(B_{m,k_{1}} \cdots B_{m,k_{N}}) \ket{k_{1}} \cdots \ket{k_{N}}.$$ 
This representation is called the \emph{irreducible form} in~\cite{de2017irreducible} and it can be effectively computed. 
Therefore, $\ket{\psi_{N}} = 0$ if and only if $\ket{\phi_{N,m}} = 0$ for all $m$. 
Without loss of generality, from now on, we always assume that the set $\setnocond{A_{k}}$'s corresponds to an irreducible QMC $(\h_{D}, \e)$ with $\e(\functiondot) = \sum_{k} A_{k} \functiondot A_{k}^{\dag}$. 
Further, $\ket{\psi_{N}} = 0$ if and only if $\braket{\psi_{N}}{\psi_{N}} = 0$. 
By simple calculations, we have 
\[
    \braket{\psi_{N}}{\psi_{N}} = \tr([\sum_{k}E_k^*\otimes E_k]^N)= \tr([M_{\e}^{N}]^*) = \tr(M_{\e}^N)
\]
where $E^{*}$ stands for the (entry-wise) complex conjugate of $E$, $M_{\e}=\sum_{k}E_k\otimes E_k^*$ is called the \emph{matrix representation} of $\e$, and the last equality in the above chain follows from $\tr(M_{\e})$ being a real number for any $\e$.

To specify the validity of the hypothesis about the ground states of 1-dimensional quantum many-body systems, we need the following kind of atomic propositions: 

\begin{definition}\label{def-sup}\begin{enumerate}\item An atomic proposition is defined to be an interval $\i\subseteq \reals$. \item A super-operator $\e\in\sh$ satisfies $a := \i$, written $\e\models a$, if the trace of its matrix representation $M_\e$  lies in interval $\i$; that is, $\tr(M_\e)\in\i$. 
\end{enumerate}\end{definition}

Similar to Section~\ref{sec:interpretstatistics}, the satisfaction relation $\e\models a$ can also be extended to $\g\models\vp$ between a QMC $\g$ and an LTL formula $\vp$. 
Here we use the labeling function: 
\begin{equation}\label{l2-fun}\lfdy \colon \sh \to 2^{\ap},\qquad \lfdy(\e) = \setcond{a \in \ap}{\e\models a}\end{equation} 
which  assigns to each super-operator the set of atomic propositions in $\ap$ satisfied by it. Furthermore, let $\lfdy(\e_1, \e_2 ,\ldots) = \lfdy(\e_1), \lfdy(\e_2), \ldots$ for any sequence of super-operators $\e_1,\e_2,\ldots$. Therefore, we define: 
$$\g\mddy\vp\ {\rm if\ and\ only\ if}\ \lfdy(\trjdy(\g))\in\mathcal{L}_\omega(\vp)$$ where $\trjdy(\g)$ is the super-operator  trajectory of $\g$ as defined in Eq. (\ref{s-traj}).

\begin{example}\label{Exa:many-body}  Given a finite set of matrices $\{A_k\}$ on a Hilbert space $\h$ corresponding to a (irreducible) QMC $\g=(\h,\e)$ with $\e(\functiondot)=\sum_k A_k\functiondot A_k^\dagger$, we set $\ap = \setnocond{\i_1, \i_2}$, where 
$\i_{1} = [0,0]$ and $\i_{2} = (-\infty,0) \cup (0,\infty)$. 
The atomic proposition $ \i_1$ (resp. $\i_2$) asserts that the trace of the matrix representation of the current super-operator is zero (resp. nonzero).  The properties  considered in Problem~\ref{pro:many_body} can be written as the LTL formulas:  
\begin{itemize}
	\item Is $\g \mddy \bigcirc^{(J)}  \square \i_2$?
	\item  Is $\g \mddy \bigcirc^{(J)} \lozenge \square \i_2$?
	 
\end{itemize}

\end{example}
\section{Model Checking $\g\mdst\vp$}
\label{sec:modelchecking}

With the notations in Eq.~(\ref{l1-fun}), the model checking problem for $\trjst(\g)$ against LTL formulas can be formally defined.
\begin{problem}
\label{prob:quantumstate}
    Given a QMC $\g = (\h, \e, \rho_{0})$, a labeling function $\lfst$, and an LTL formula $\vp$, decide whether $\g \mdst \vp$, i.e., whether $\lfst(\trjst(\g)) \in \l_\omega(\vp)$.
\end{problem}

As QMCs can simulate classical Markov chains (see Appendix~\ref{More-QMC}), the counter-example presented in~\cite{agrawal2015approximate} can be used to show that the language $\setnocond{\lfst(\trjst(\g))}$ is generally not $\omega$-regular.
Thus the standard approach of model checking $\omega$-regular languages is not directly applicable to solve Problem~\ref{prob:quantumstate}.
Following the techniques introduced in~\cite{agrawal2015approximate}, we turn to consider approximate verification problems of QMCs.
To this end, we introduce the notions of neighborhoods for quantum states and for sequences of quantum states, which are the tasks of the following two subsections. 

For simplicity, in this section, we write $\g\mdst\vp$ as $\g\models\vp$. All proofs for the results presented in this section can be found in Appendix~\ref{sec:periodicstable}.

\subsection{Neighborhood of quantum states}
 The definition of neighborhood of quantum states is induced by vector norms on $\bh$, so we first recall the vectorization of quantum states. 

Given a super-operator $\e = \{E_{k}\}_k$, its matrix representation 
 $M_{\e}$ is a linear operator on  $\h \otimes \h$ and furthermore, for any $A \in \bh$,
\[
    \e(A) = \tr_{2}(M_{\e}(A \otimes I) \ket{\Omega} \bra{\Omega})
\]
where $\tr_{2}$ is the partial trace on the second Hilbert space and $\ket{\Omega}$ is the unnormalized maximally entangled state in $\h \otimes \h$, i.e., $\ket{\Omega} = \sum_{j} \ket{j} \otimes \ket{j}$ with an orthonormal basis $\setnocond{\ket{j}}$ of $\h$. 
As a simple consequence, for the composition $\e_{2} \circ \e_{1}$ of super-operators where for any $A \in \bh$, $(\e_{2} \circ \e_{1})(A) = \e_{2}(\e_{1}(A))$, $M_{\e_{2} \circ \e_{1}}$ is exactly the matrix product $M_{\e_{2}} M_{\e_{1}}$. 
For simplicity, in this paper we freely interchange $\e$ and $M_\e$.

Finally, note that any linear map $T$ on $\bh$ admits up to $\dim(\h)^{2}$ complex \emph{eigenvalues} $\lambda$ satisfying $T(A) = \lambda A$ for some $0 \neq A \in \bh$. 
We write $\spec(T)$ for the spectrum of $T$, i.e., the set of all eigenvalues of $T$. 
The \emph{spectral radius} of $T$ is defined as $\radius{T} = \max\setcond{|\lambda|}{\lambda \in \spec(T)}$. 
In particular, for any super-operator $\e \in \sh$, $\radius{\e} = 1$. 
Denote by
\[
    \maxmag{T} = \setcond{\lambda \in \spec(T)}{|\lambda| = \radius{T}}
\]
the set of eigenvalues of $T$ with maximal magnitude. 
Note that the calculation of $\spec(\e)$ boils down to that of $\spec(M_{\e})$ since
\[
    \text{$\e(A) = \lambda A$ \quad if and only if \quad $M_{\e} \ket{A} = \lambda \ket{A}$,}
\]
where $\ket{A}$ is the vectorization of $A \in \bh$, i.e., $\ket{A} = (A \otimes I) \ket{\Omega}$. 

We choose to use the vector norm on $\bh$ and the induced operator norm on $\bbh$. 
The result will apply for any other norm, as all norms on a finite-dimensional Hilbert space are equivalent~\cite{Horn2013}.
\begin{definition} 
    Given a Hilbert space $\h$ with $d = \dim(\h)$, 
    \begin{itemize}
    \item 
        the vector norm $\norm{A}$ of $A \in \bh$ is defined to be the norm of the vector $\ket{A}$, that is, $\norm{A} \coloneqq |\braket{A}{A}|$;
    \item 
        the operator norm $\opnorm{\functiondot}$ on $\bbh$ induced by $\norm{\functiondot}$ is
        \[
            \opnorm{T} \coloneqq \sup \setcond{\norm{T(A)}}{\text{$A \in \bh$ with $\norm{A} = 1$}}.
        \]
    \end{itemize}
\end{definition}

For convenience, we denote $\opnorm{\functiondot}$ by $\norm{\functiondot}$ if no confusion arises. 
One can easily show that for any $\e \in \sh$,
\[
    \norm{\e} = \norm{M_{\e}} = \max_{\lambda \in \spec(M_{\e}^{\dag} M_{\e})} \sqrt{\lambda}. 
\]
That is, $\norm{\e}$ is the maximum singular value of $M_{\e}$. 
Then by the above equation, for any $\e_{1}, \e_{2}, \f \in \sh$,
\begin{equation}
    \norm{(\e_{1} - \e_{2}) \circ \f} \leq \norm{\e_{1} - \e_{2}}.
    \label{eq:contractive}
\end{equation}

Furthermore, for any $\rho \in \dh$ and $\e_{1}, \e_{2} \in \sh$,
\begin{eqnarray}
    \norm{\e_{1}(\rho) - \e_{2}(\rho)} \leq \norm{\e_{1} - \e_{2}} \cdot \norm{\rho} \leq \norm{\e_{1} - \e_{2}},
    \label{eq:operatornorm}
\end{eqnarray}
where the second inequality follows from $\norm{\rho} \leq 1$.

With these norms, $\epsilon$-neighborhood of quantum states can be defined as follows:
\begin{definition}
    Given a density operator $\rho \in \dh$ and $\epsilon > 0$, the (symbolic) $\epsilon$-neighborhood $\epsneigh(\rho)$ of $\rho$ is defined to be a subset of $2^{\ap}$:
    \[
        \epsneigh(\rho) = \setcond{\lfst(\rho')}{\rho' \in \dh, \norm{\rho - \rho'} < \epsilon}.
    \]
\end{definition}



\subsection{Neighborhoods of trajectories of QMCs}

A key property held by classical Markov chains, which plays an essential role in the approximate verification techniques developed in~\cite{agrawal2015approximate},
is the following.
\begin{proposition}[cf.~\cite{agrawal2015approximate}]
\label{prop:pstable}
    For any Markov chain $(\s, P, \mu_{0})$, there is an integer $\theta > 0$ such that
    \[
        \lim_{n \to \infty} \mu_{0} P^{n \theta} = \mu^{*}
    \]
    for some limiting distribution $\mu^{*}$. 
    Furthermore, $\theta$ is independent of $\mu_{0}$.
\end{proposition}

For QMCs, we can define a similar notion.
\begin{definition}
\label{def:e_periodicstable}
    A QMC $\g = (\h, \e, \rho_{0})$ is called \emph{periodically stable} if there exists an integer $\theta > 0$ such that 
    \[
        \lim_{n \to \infty} \e^{n \theta}(\rho_{0}) = \rho^{*}
    \]
    for some limiting quantum state $\rho^{*}$. 
    The minimal such $\theta$, if it exists, is called the \emph{period} of $\g$ and it is denoted as $p(\g)$.
\end{definition}

Proposition~\ref{prop:pstable} essentially says that all classical Markov chains are periodically stable. 
However, as the following example shows, such a property does not hold for QMCs. 
\begin{example}
\label{ex:u}
    Let $\h$ be a two-dimensional Hilbert space with $\setnocond{\ket{0}, \ket{1}}$ being an orthonormal basis of it. 
    Let $U = \ket{0} \bra{0} + e^{i2\pi\psi} \ket{1} \bra{1}$ be a unitary operator on $\h$ where $\psi$ is irrational. 
    Then for $\rho_{0} \in \dh$, we can easily show that the QMC $(\h, \e_{U}, \rho_{0})$ where $\e_{U}(\rho) = U \rho U^{\dag}$ is not periodically stable in general.

    In fact, by a simple calculation we know 
    \[
        \e_{U}^{n \theta}(\rho_{0}) = \rho_{00} \cdot \ket{0} \bra{0} + \rho_{11} \cdot \ket{1} \bra{1} + e^{-i 2 \pi \psi n \theta} \rho_{01} \cdot \ket{0} \bra{1} + e^{i 2 \pi \psi n \theta} \rho_{10} \cdot \ket{1} \bra{0}
    \]
    where $\rho_{ij} = \bra{i} \rho \ket{j}$.
    Note that as $\psi$ is irrational, the set $\setcond{e^{i 2 \pi \psi m}}{m \in \naturals}$ is dense in the unit circle~\cite{hardy1979introduction}. 
    Thus for any integer $\theta > 0$, the limit $\lim_{n \to \infty} \e_{U}^{n \theta}(\rho_{0})$ cannot exist, unless $\rho_{01} = \rho_{10} = 0$. 
\end{example}



Note that $\e_{U}$ in Example~\ref{ex:u} has four eigenvalues (counting multiplicity) $1$, $1$, $e^{-i 2 \pi \psi}$, and $e^{i 2 \pi \psi}$, with the corresponding eigenvectors $\ket{0} \bra{0}$, $\ket{1} \bra{1}$, $\ket{0} \bra{1}$, and $ \ket{1} \bra{0}$, respectively. 
We have shown that $(\h, \e_{U}, \rho_{0})$ is periodically stable if and only if $\rho_{0}$ vanishes in the directions of $\ket{0} \bra{1}$ and $ \ket{1} \bra{0}$.
Interestingly, this is the exact reason for a QMC not to be periodically stable (see Appendix~\ref{sec:periodicstable}). 
That is, a QMC is periodically stable if and only if the initial quantum state has no components in the directions determined by eigenvectors of the relevant super-operator corresponding to eigenvalues of the form $e^{i 2 \pi \psi}$ for some irrational number $\psi$. 
This result also provides us with an efficient way to check if a given QMC is periodically stable (and so the technique of approximate verification developed in this paper applies).

Now we focus on periodically stable QMCs, and present our key lemma (see Appendix~\ref{sec:periodicstable} for the definition of the special super-operator $\e_{\phi}$).  
\begin{lemma}
\label{mainlemmaQMC}
    Given a periodically stable QMC $\g = (\h, \e, \rho_{0})$ with period $p(\g)$, let $\eta_{k} = \e_{\phi}(\e^{k}(\rho_{0}))$, for each $0 \leq k < p(\g)$.
    Then for any $\epsilon > 0$, there exists an integer $K^{\epsilon} > 0$ such that for any $n > K^{\epsilon}$,
    \[
        \lfst(\e^{n}(\rho_{0})) \in \epsneigh(\eta_{n \modulo p(\g)}).
    \]
    Furthermore, the time complexities of computing $p(\g)$ and $K^{\epsilon}$ are both in $\bigO{d^{8}}$, where $d = \dim(\h)$.
\end{lemma}

With Lemma~\ref{mainlemmaQMC}, we can define the notion of neighborhood of trajectories for periodically stable QMCs.

\begin{definition}
\label{def:trajneigh}
    Given a periodically stable QMC $\g= (\h, \e, \rho_{0})$ and $\epsilon > 0$, the (symbolic) $\epsilon$-neighborhood of the trajectory $\trjst(\g)$ of $\g$ is defined to be the language $\epsneigh(\trjst(\g))$ over $(2^{\ap})^{\omega}$ such that $\xi \in \epsneigh(\trjst(\g))$ if and only if
    \begin{itemize}
    \item 
        $\xi[n] = \lfst(\e^{n}(\rho_{0}))$ for all $0 \leq n \leq K^{\epsilon}$;
    \item 
        $\xi[n] \in \epsneigh(\eta_{n \modulo p(\g)})$ for all $n > K^{\epsilon}$,
    \end{itemize}
    where the states $\setcond{\eta_{k}}{0 \leq k <  p(\g)}$ and $K^{\epsilon}$ are as given in Lemma~\ref{mainlemmaQMC}.
\end{definition}

\subsection{Approximate Verification of QMCs}

With Definition~\ref{def:trajneigh}, we can state and solve the approximate model checking problems for QMCs against LTL formulas as follows.

\begin{problem}
\label{prob:approx}
    Given a periodically stable QMC $\g = (\h, \e, \rho_{0})$, a labeling function $L$, an LTL formula  $\vp$, and $\epsilon > 0$, decide whether
    \begin{enumerate}
    \item 
         \emph{$\g$ $\epsilon$-approximately satisfies $\vp$ from below}, denoted $\g \models_{\epsilon} \vp$; that is, whether $\epsneigh(\trjst(\g)) \cap \l_{\omega}(\vp) \neq \emptyset$; 
    \item 
         \emph{$\g$ $\epsilon$-approximately satisfies $\vp$ from above}, denoted $\g \models^{\epsilon} \vp$; that is, whether $\epsneigh(\trjst(\g)) \subseteq \l_{\omega}(\vp)$. 
    \end{enumerate}
\end{problem}

To justify that Problem~\ref{prob:approx} is indeed an approximate version of Problem~\ref{prob:quantumstate}, we first note that $\lfst(\trjst(\g)) \in \epsneigh(\trjst(\g))$. 
Then we have three cases:
\begin{enumerate}
\item 
    if $\g \not\models_{\epsilon} \vp$, then $\epsneigh(\trjst(\g)) \cap \l_{\omega}(\vp)= \emptyset$, and hence $\lfst(\trjst(\g)) \notin \l_{\omega}(\vp)$;
\item 
    if $\g \models^{\epsilon} \vp$, then $\epsneigh(\trjst(\g)) \subseteq \l_{\omega}(\vp)$, and hence $\lfst(\trjst(\g)) \in \l_{\omega}(\vp)$;
\item 
    if neither $\g \not\models_{\epsilon} \vp$ nor $\g \models^{\epsilon} \vp$, then we may halve $\epsilon$ and repeat the approximate model checking procedure presented in cases 1 and 2. 
\end{enumerate}

The first two cases both give (negative or affirmative) answers to Problem~\ref{prob:quantumstate}. 
Note that in some extreme situation, the procedure presented above may not terminate. 
To determine when the procedure terminates seems difficult and we would like to leave it
as future work.

Finally, to solve Problem~\ref{prob:approx}, we represent $\epsneigh(\trjst(\g))$ in Definition~\ref{def:trajneigh} as an $\omega$-regular expression
\[
    \epsneigh(\trjst(\g)) = \setnocond{\lfst(\rho_{0})} \cdot \setnocond{\lfst(\e(\rho_{0}))} \cdots \setnocond{\lfst(\e^{K^{\epsilon}}(\rho_{0}))} \cdot \left(\epsneigh(\zeta_{1}) \cdots \epsneigh(\zeta_{p(\g)})\right)^{\omega}
\]
where $\zeta_{i} = \eta_{(K^{\epsilon} + i) \modulo p(\g)}$, $1 \leq i \leq p(\g)$,
and for any two sets $X$ and $Y$, $X \cdot Y = \setcond{xy}{x \in X, y \in Y}$.
Thus $\epsneigh(\trjst(\g))$ is $\omega$-regular and standard techniques~\cite{baier2008principles,HandbookMC18} can be employed to check if $\epsneigh(\trjst(\g)) \cap \l_{\omega}(\vp) = \emptyset$ or $\epsneigh(\trj(\g)) \subseteq \l_{\omega}(\vp)$. 

\begin{theorem} 
\label{thm:main}
    Given a periodically stable QMC $\g = (\h, \e, \rho_{0})$ with $\dim(\h)=d$, a labeling function $L$, an LTL formula  $\vp$, and $\epsilon > 0$, the approximate verification problems presented in Problem~\ref{prob:approx} can be solved in time $\bigO{2^{\bigO{|\vp|}} \cdot (K^{\epsilon} + p(\g))}=\bigO{2^{\bigO{|\vp|}}d^8}$, where $|\vp|$ is the length of $\vp$. 
\end{theorem}

\begin{algorithm}
    \caption{ModelCheck($\g,  \ap, L_s, \vp, \epsilon$)}
    \label{ModelCheck}
    \begin{algorithmic}[1]
        \REQUIRE A periodically stable QMC $\g = (\h, \e, \rho_0)$ with Kraus operators $\setnocond{E_{i}}_{i}$, a finite set of atomic propositions $\ap$,  a labeling function $\lfst$,  an LTL formula $\vp$, and $\epsilon> 0$.
        \ENSURE \TRUE{}, \FALSE{}, or \textbf{unknown}, where \TRUE{} indicates $\g \models \vp$, \FALSE{} indicates $\g \not\models \vp$, and \textbf{unknown} stands for an inconclusive answer.
        \STATE Compute $M_{\e}$ as $\sum_{i} E_{i} \otimes E_{i}^{*}$
        	\label{algo:line:matrixrepresentation}
        \STATE Decompose $M_\e$ into $S$, $J$, and $S^{-1}$, i.e., $M_\e=SJS^{-1}$,  by Jordan decomposition
        \STATE Get $M_{\e_{\phi}} = S J_{\phi} S^{-1}$, $C$, and $p(\g)$ by means of $S$ and $J$
        \STATE Compute the integer $M^{\epsilon}$ by solving $C(M^{\epsilon}p(\g))^{d_{\mu}-1} < \epsilon$  and $M^{\epsilon}p(\g)+1 \geq d_{\mu} $
        \STATE Set $K^{\epsilon}$ to be $M^{\epsilon} p(\g)$
        \STATE Compute $\ket{\rho_0}$ as $(\rho_0\otimes\i)\ket{\Omega}$ 
        \FORALL{$k \in \setnocond{0, 1, \dotsc, p(\g) - 1}$}
            \STATE Set $\ket{\eta_k}$ to be $M_{\e_{\phi}}M^{k}_{\e}\ket{\rho_0}$
            \STATE Compute $\epsneigh(\eta_k)$ by semidefinite programming
        \ENDFOR
        	\label{algo:line:neighbourhood}
        \FORALL{$k \in \setnocond{0, 1, \dotsc, K^{\epsilon}-1}$}
        	\label{algo:line:prefixbegin}
        \STATE Compute $\rho_k$ as $\e^k(\rho_0)$
        \STATE Compute $\lfst(\rho_k)$
        \ENDFOR
        	\label{algo:line:prefixend}
        \STATE Let $\epsneigh(\trj(\g))$ be the $\omega$-regular language\\
            ~\hfill$\epsneigh(\trj(\g)) = \setnocond{\lfst(\rho_{0})} \setnocond{\lfst(\rho_{1})} \cdots \setnocond{\lfst(\rho_{K^{\epsilon}-1})} \cdot (\epsneigh(\eta_{0}) \epsneigh(\eta_{1}) \cdots \epsneigh(\eta_{p(\g)-1}))^{\omega}$\hfill~
        \STATE Construct the NBA $A_{\vp}$ for $\vp$ \COMMENT{standard construction}
            \label{algo:line:NBAformula}
        \STATE Construct the NBA $A_{U}$ accepting $\epsneigh(\trj(\g))$\COMMENT{standard lasso-shaped construction}
        \label{algo:line:NBAlasso}
        \IF[standard B\"uchi automata operations]{$\l(A_{U}) \cap \l(A_{\vp}) = \emptyset$}
            \label{algo:line:NBAintersectionEmptiness}
            \RETURN \FALSE
        \ELSIF[standard B\"uchi automata operations]{$\l(A_{U}) \subseteq \l(A_{\vp})$}
            \label{algo:line:NBAinclusion}
            \RETURN \TRUE
        \ELSE
            \RETURN \textbf{unknown}
        \ENDIF
    \end{algorithmic}
\end{algorithm}
In the end, we develop a model checking algorithm (Algorithm~\ref{ModelCheck}) to answer Problem~\ref{prob:approx}. From line~\ref{algo:line:matrixrepresentation} to line~\ref{algo:line:neighbourhood}, by Lemma~\ref{mainlemmaQMC}, we compute $K^\epsilon$, $p(\g)$, matrix representation $\{\eta_{k}\}_{k=0}^{p(\g)-1}$ and $\{U_{\epsilon}(\eta_{k})\}_{k=0}^{p(\g)-1}$. The aim of the processing line~\ref{algo:line:prefixbegin} to line~\ref{algo:line:prefixend} is to calculate $\{\lfst(\rho_k)\}_{k=0}^{K^\epsilon-1}$. 
The B\"uchi automaton $A_{\vp}$ for the LTL formula $\vp$ is constructed at line~\ref{algo:line:NBAformula} by means of a standard construction (see, e.g.,~\cite{HandbookMC18}) while the B\"uchi automaton $A_{U}$ at line~\ref{algo:line:NBAlasso} is obtained by an ordinary lasso-shaped construction: 
it is enough to insert a new state between each letter, make the state joining the stem and the lasso part accepting, and use the accepting state as the target of the last action in the lasso.
The two operations on B\"uchi automata at lines~\ref{algo:line:NBAintersectionEmptiness} and~\ref{algo:line:NBAinclusion} are standard operations:
intersection and emptiness reduce to automata product and strongly connected components decomposition, which require quadratic time (cf.~\cite{HandbookMC18}). Language inclusion, however, in general requires exponential time and is PSPACE-complete (cf.~\cite{HandbookMC18});
in our case, however, we can remain in quadratic time by replacing the check $\l(A_{U}) \subseteq \l(A_{\vp})$ with the check $\l(A_{U}) \cap \l(A_{\neg\vp}) = \emptyset$, since it is common in the model checking community to assume that constructing the B\"uchi automata $A_{\vp}$ and $A_{\neg\vp}$ require the same effort.

\section{Modeling Checking $\g\mddy \vp$}\label{Sec:sigma2}
 With the notations in Eq.~(\ref{l2-fun}), the model checking problem for $\trjdy(\g)$ against LTL formulas can be formally defined.
\begin{problem}
\label{prob:superoperator}
    Given a QMC $\g = (\h, \e)$ with $\e$ being periodically stable, a labeling function $\lfdy$, and an LTL formula $\vp$, decide whether $\g \mddy \vp$, i.e., whether $\lfdy(\trjdy(\g)) \in \l_\omega(\vp)$.
\end{problem}

A super-operator $\e$ is called \emph{periodically stable} if there exists an integer $\theta > 0$ such that $\lim_{n \to \infty} \e^{n \theta}$ exists in $\sh$. 
The minimal such $\theta$, if it exists, is called the \emph{period} of $\e$ and denoted by $p(\e)$.
Similar to model checking $\g\mdst \vp$, we can define $U_{\epsilon}(\e)$ and $U_{\epsilon}(\trjdy(\g))$ for any $\e \in \sh$ and $\epsilon > 0$. In this section, we simply write $\g\mddy \vp$ as $\g\models \vp$. We hope to answer the following approximate model checking question:
\begin{problem}
\label{prob:quantum}
    Given a QMC $\g=(\h,\e)$ with $\e$ being periodically stable, a labeling function $\lfdy \colon \sh \to 2^{\ap}$, and an LTL formula $\vp$, decide whether 
    \begin{enumerate}
    \item 
         \emph{$\g$ $\epsilon$-approximately satisfies $\vp$ from below}, denoted $\g \models_{\epsilon} \vp$; that is, whether $\epsneigh(\trjdy(\g)) \cap \l_{\omega}(\vp) \neq \emptyset$; 
    \item 
         \emph{$\g$ $\epsilon$-approximately satisfies $\vp$ from above}, denoted $\g \models^{\epsilon} \vp$; that is, whether $\epsneigh(\trjdy(\g)) \subseteq \l_{\omega}(\vp)$. 
    \end{enumerate}
\end{problem}
It turns out that Problem~\ref{prob:quantum} can be easily reduced to Problem~\ref{prob:approx} (see Appendix~\ref{Sec:reduction} for the proof), so Algorithm~\ref{ModelCheck} can be directly used to solve Problem~\ref{prob:quantum}. However, the complexity increases significantly. So we also develop a direct method for it (see Appendix~\ref{Section:solvemanybody} for the details). 

\begin{theorem}
\label{thm:approxs}
    Given a  QMC $\g=(\h,\e)$ with $\e$ being periodically stable and dim$(\h)=d$, a labeling function $L$, an LTL formula $\vp$, and $\epsilon > 0$, the following two problems can be solved in time $\bigO{2^{\bigO{|\vp|}} \cdot (K^{\epsilon} + p(\e))}=\bigO{2^{\bigO{|\vp|}}d^8}:$
    \begin{enumerate}
    \item 
        decide whether \emph{$\g$ $\epsilon$-approximately satisfies $\vp$ from below}, denoted $\g \models_{\epsilon} \vp$; that is, to check whether $\epsneigh(\trjdy(\g))\cap \l_{\omega}(\vp)\neq \emptyset$;
    \item 
        decide whether  \emph{$\g$ $\epsilon$-approximately satisfies $\vp$ from above}, denoted $\g \models^{\epsilon} \vp$; that is, to check whether $\epsneigh(\trjdy(\g))\subseteq \l_{\omega}(\vp)$.
    \end{enumerate}
\end{theorem}

It is worth noting that as irreducible QMCs are the QMCs with periodically stable super-operators (see Appendix~\ref{Section:solvemanybody}), the quantum many-body problems in Problem~\ref{pro:many_body} can always be approximately answered by the reduction processes from the general case to the irreducible case in Section~\ref{Subsec:many-body} and Theorem~\ref{thm:approxs}.

\section{Experiments}
\label{sec:experiments}

In this section, to show the use of model checking techniques developed in this paper, we run experiments on 
AKLT (Affleck-Kennedy-Lieb-Tasaki) and cluster models, which are two essential 1-dimensional quantum many-body systems~\cite{Perez-Garcia2007}.
All source codes have been submitted as supplemental materials.
\subsection{AKLT model}

The AKLT model was introduced by Affleck, Kennedy, Lieb and Tasaki in~\cite{affleck1988valence}, and it was the first analytical example of a quantum spin chain supporting the so-called Haldane’s conjecture: it is a local spin-1 Hamiltonian with Heisenberg-like interactions and a non-vanishing spin gap in thermodynamic limit. The spin of a particle describes its possible angular momentum values. 

The ground state of AKLT model can be expressed by the MPS as follows:
\[\ket{\psi_{N}} = \sum_{k_{1},\dotsc,k_{N}=1}^{3} \tr(A_{k_{1}} \cdots A_{k_{N}}) \ket{k_{1}}\otimes \cdots \otimes\ket{k_{N}}
\]
where
\[
    A_{1} =
        \left[
            \begin{matrix}
                0 & \sqrt{2/3} \\
                0 & 0
            \end{matrix}
        \right],
    \qquad
    A_{2} = 
        \left[
            \begin{matrix}
                -\sqrt{1/3} & 0 \\
                0 & \sqrt{1/3}
            \end{matrix}
        \right],
    \qquad
    A_{3} = 
        \left[
            \begin{matrix}
                0 & 0 \\
                -\sqrt{2/3} & 0 
            \end{matrix}
        \right].
\]

Note that the set $\{A_{k}\}_{k=1}^3$ of matrices corresponds to  an irreducible QMC, and so the corresponding super-operator is periodically stable (see Appendix~\ref{Section:solvemanybody}). Therefore, by Theorem~\ref{thm:approxs} the validity of the MPS in Problem~\ref{pro:many_body} can be approximately checked. Specifically, using the notations in Example~\ref{Exa:many-body} and setting $J=3$, we can check whether $\g\models_d \bigcirc^{(J)} \square \i_2$ and $\g\models_d \bigcirc^{(J)} \lozenge \square \i_2$ by implementing the reduction in Appendix~\ref{Sec:reduction} and  Algorithm~\ref{ModelCheck}, where $\g=(\h,\e)$, dim$(\h)=2$, and $\e=\{A_{k}\}_{k=1}^3$.

We repeatedly run the algorithm by setting $\epsilon=0.5$ initially and halving it in the next iteration whenever $\textbf{unknown}$ is returned. After three iterations, we get the answer $\textbf{true}$, which indicates that the MPS of ground states of the AKLT model is valid in thermodynamic limit.
The detailed result is shown in Table~{\ref{tb:experiment}.

\subsection{Cluster Model}

\begin{table}[!tbp]
\centering
\renewcommand{\arraystretch}{1.3}
\begin{tabular}{|c|c|c|c|}
\hline
\diagbox{$\vp$}{Output}{$\epsilon$} & 0.5 & 0.25 & 0.125 \\ 
\hline
$\bigcirc^{(J)} \square \i_2$&\textbf{unknown}&\textbf{unknown}&\textbf{true}\\
\hline
$\bigcirc^{(J)} \lozenge \square \i_2$&\textbf{unknown}&\textbf{unknown}&\textbf{true}\\
\hline
\end{tabular}
\vspace{1em}
\caption{The experimental results of checking ground states of both AKLT and cluster models}
\label{tb:experiment}
\end{table}
A state in cluster models is a type of highly entangled state of multiple particles~\cite{briegel2001persistent}, and has been realized experimentally. It is generated in lattices of particles with Ising type interactions~\cite{zeng2015quantum}, and is especially used as a resource state for the one-way quantum computation~\cite{onewaycomputer}.

The ground state of a cluster model can be expressed by the MPS as follows:
\[\ket{\psi_{N}} = \sum_{k_{1},\dotsc,k_{N}=1}^{2} \tr(A_{k_{1}} \cdots A_{k_{N}}) \ket{k_{1}}\otimes \cdots \otimes\ket{k_{N}}
\]
where
\[
    A_1 =
        \frac{1}{\sqrt{2}}
        \left[
            \begin{matrix}
                0 & 0 \\
                1 & 1
            \end{matrix}
        \right],
    \qquad
    A_2 = 
        \frac{1}{\sqrt{2}}
        \left[
            \begin{matrix}
                1 & -1 \\
                0 & 0
            \end{matrix}
        \right].
    \qquad
\]

Similar to AKLT model (as the set $\{A_{k}\}_{k=1}^2$ of matrices  also corresponds to an irreducible QMC), Algorithm~\ref{ModelCheck} can be used to check the validity of the above MPS. The experimental result is the same as that in Table 1, from which we conclude that the MPS of ground states of the cluster model is also valid in thermodynamic limit.

\section{Conclusion}  

In this paper, we show that
model checking techniques can be adapted for applications in quantum statistical mechanics and quantum many-body systems. The key observation is that the evolution of quantum systems in a question can be modeled by a QMC and the properties that interest us can be described by appropriate LTL formulas. Interestingly, by interpreting the dynamics of QMCs and the atomic propositions of LTL in different ways, the same model checking technique can be used for different applications; for details, see Problem~\ref{pro:statistic} for long-term classical information in quantum systems and Problem~\ref{pro:many_body} for validity of MPS ground states in many-body systems.
We then present an effective algorithm to approximately model check  a  periodically stable QMC against LTL formulas. Examples from AKLT and cluster models, two important 1-dimensional quantum many-body systems, are studied to illustrate the utility of our algorithm.
 
For future study, we are going to develop model checking algorithms for QMCs which are not periodically stable. Note that by Proposition~\ref{prop:pstable}, these QMCs have no classical counterparts, and novel techniques must be invented to analyze their long-term behaviors. Another interesting line of research is to find more applications of our model checking techniques in other research fields such as quantum algorithm analysis and quantum programming theory~\cite{ying2016foundations}.
    
\section*{Acknowledgments} 
This work was partly supported by the National Key R$\&$D Program of China (Grant No: 2018YFA0306701), the National Natural Science Foundation of China (Grant No: 61832015) and the Australian Research Council (ARC) under grant Nos. DP160101652 and DP180100691.
%
%

\bibliographystyle{abbrv}
\bibliography{note110418}
\newpage
\section*{Appendices}

\appendix
  \renewcommand{\appendixname}{Appendix~\Alph{section}}

\section{More about Quantum Markov Chains}\label{More-QMC}

QMCs offer an exceptional paradigm for modeling the evolution of quantum systems: 
they 
were first introduced as a model of quantum communicating systems~\cite{wolf2012quantum}, while quantum random walks, a special class of QMCs, have been successfully used to design quantum algorithms (see~\cite{ambainis2003quantum,kempe2003quantum} for a survey of this research line). 
More recently, QMCs were used as a quantum memory model for preserving quantum states~\cite{guan2018decomposition,guan2017super} and as a semantic model for the purpose of verification and termination analysis of quantum programs~\cite{ying2013verification,yu2012reachability,ying2013reachability}. 
$\g=(\h,\e,\rho_0)$ is the most general quantum Markov chain and there also emerged some special cases, such as open quantum random walks~\cite{Attal2012} and classical-quantum Markov chains~\cite{feng2013model}, the latter being classical MCs where the transition probability matrix is replaced by a transition super-operator matrix. 
For studying the dynamical properties, some researchers contributed some interesting results case-by-case. 
For the long-run behavior, \cite{wolf2012quantum,guan2018decomposition} gave some characterizations for limiting states; 
for reachability probabilities, some decomposition techniques of quantum Markov chains in terms of irreducibility and periodicity were obtained in~\cite{ying2013reachability,Baumgartner2011,guan2018decomposition}, and further the reachability, repeated reachability, and persistence probabilities were computed. 

Model checking quantum systems has been studied in the last 10 years, with the main purpose of verifying quantum communication protocols; for the details, we refer to the review paper~\cite{ying2018model}. 
Recently, some researchers considered model checking quantum automata and developed an algorithm for checking linear-time properties (e.g., invariants and safety properties) in~\cite{ying2014model}. 
Following Birkhoff-von Neumann quantum logic~\cite{birkhoff1936logic}, they used closed subspaces of Hilbert spaces as the atomic propositions about the state of the system, and the specifications were represented by infinite sequences of sets of atomic propositions. 
After that, by adding a classical graph, a special quantum Markov chain, called super-operator-valued quantum Markov chains, was proposed in~\cite{feng2013model} for modeling quantum programs and quantum cryptographic protocols. 
Furthermore, a quantum extension of probabilistic computation tree logic (PCTL) was defined and a model-checking algorithm for this Markov model was developed.

In the following, we present several examples of QMCs. 

\begin{example}Not surprisingly, any classical Markov chain $(\s, P, \mu_{0})$ can be effectively encoded as a QMC:
define $\h$ to be a $|\s|$-dimensional Hilbert space spanned by an orthonormal basis $\setcond{\ket{s}}{s \in \s}$ and $\e$ be a super-operator with Kraus operators 
\[
    \setcond{E_{s,t} = \sqrt{P(s,t)} \ket{t} \bra{s}}{s, t \in \s}.
\]
It is easy to check that $\e$ is completely positive and trace-preserving. 
Furthermore, let $\rho_{0} = \sum_{s \in \s} \mu_{0}(s) \ket{s} \bra{s}$. 
Then the QMC $(\h, \e, \rho_{0})$ fully simulates the behavior of $(\s, P, \mu_{0})$ in the sense that for all $n \geq 0$,
\[
    \e^{n}(\rho_{0}) = \sum_{s \in \s} \mu_{n}(s) \ket{s} \bra{s}
\]
where $\mu_{n} = \mu_{0} P^{n}$, $\e^{0} = \mathrm{id}_{\h}$, and $P^{0} = I$.\end{example}
\begin{example}[Amplitude-damping channel]
Consider the 2-dimensional amplitude-damping channel modeling the physical processes such as spontaneous emission. Let  $\mathcal{H}=\textrm{span}\{|0\rangle, |1\rangle\}$, and
\begin{eqnarray*}
\mathcal{E}(\rho)=E_0 \rho E_0^\dag + E_1 \rho E_1^\dag
\end{eqnarray*}
where $E_{0}=|0\rangle\langle 0|+\sqrt{1-p}|1\rangle\langle 1|$ and $E_{1}=\sqrt{p} |0\rangle\langle 1|$ with $p>0$.
\end{example}
\begin{example}
\label{reducible}
Consider a natural way to encode the classical NOT gate $X: 0\rightarrow 1; 1\rightarrow 0$ into a quantum operation. Let $\mathcal{H}=\textrm{span}\{|0\rangle, |1\rangle\}$. The
super-operator $\mathcal{E}: D(\mathcal{H}) \rightarrow D(\mathcal{H})$ is defined by 
\begin{eqnarray*}
\mathcal{E}(\rho)=|1\rangle\langle 0|\rho|0\rangle \langle 1|+|0\rangle\langle 1|\rho|1\rangle \langle 0|
\end{eqnarray*}
for any $\rho \in D(\mathcal{H})$. It is easy to check that the quantum Markov chain $(\mathcal{H},\mathcal{E})$ is irreducible.
\end{example}

\section{Periodical Stability of Quantum Markov Chains}
\label{sec:periodicstable}

In this section, we give an easily checkable characterization of periodical stability of QMCs, and complete the proof of Lemma~\ref{mainlemmaQMC}.


Let $(\h, \e, \rho_{0})$ be a QMC with $M_{\e}$ being the matrix representation of $\e$ and $M_{\e} = SJS^{-1}$ its Jordan decomposition. 
Furthermore, 
\begin{equation}
	J = \bigoplus_{k=1}^{K} J_{k}(\lambda_{k}) \quad \text{and} \quad J_{k}(\lambda_{k}) = \lambda_{k} P_{k} + N_{k},
	\label{eq:Jordan}
\end{equation}
where $\lambda_{k} \in \spec(\e)$, $P_{k}$ is a projector, and $N_{k}$ the corresponding nilpotent part. 
Note that Jordan decomposition is not unique,  and we define 
\[
\alpha(\e) = \inf_{S} \{\norm{S} \cdot \norm{S^{-1}} \ :\ S^{-1}M_\e S \mbox{ is in Jordan normal form}\}
\]
 to be the \emph{Jordan condition number}~\cite{wolf2012quantum} of $\e$. 
From~\cite[Proposition 6.2]{wolf2012quantum}, the geometric multiplicity of any $\lambda_{k} \in \maxmag{\e}$ equals its algebraic multiplicity, i.e., $N_{k} = 0$.
We define 
\[
	J_{\phi} = \bigoplus_{k : |\lambda_{k}| = 1} P_{k}
\]
to be the projector onto the eigenspace corresponding to eigenvalues in $\maxmag{\e}$.
By~\cite[Proposition 6.3]{wolf2012quantum}, $SJ_{\phi} S^{-1}$ is indeed the matrix representation of some super-operator $\e_{\phi}$. 
One of the essential results regarding $\e_{\phi}$ is the following lemma from~\cite{wolf2012quantum} when we consider asymptotic properties of $\e$.
\begin{lemma}[{cf.~\cite[Theorem 8.23]{wolf2012quantum}}]
\label{lem:inequality}
	For $n \in \naturals^{+}$ we have
	\[
		C^{-1}\mu^{n} n^{d_{\mu} - 1} \leq \norm{M_{\e}^{n} - M_{\e_{\psi}}^{n}} \leq C\mu^{n} n^{d_{\mu}-1}
	\]
	where $\e_{\psi} = \e \circ \e_{\phi}$, $C$ is a positive constant determined by the Jordan decomposition of $M_{\e}$, $\mu = \sup\setcond{|\lambda|}{\lambda \in \spec(\e), |\lambda| <1}$ is the largest modulus of eigenvalues of $\e$ in the interior of the unit disc, and $d_{\mu}$ is the dimension  of the largest Jordan block corresponding to  eigenvectors of modulus $\mu$. 
	Specifically, 
	\[
	     C = 
			\begin{cases}
				\alpha(\e) & \text{if $d_{\mu} = 1$,}\\
				\alpha(\e)(\mu(d_{\mu} - 1))^{ d_{\mu}-1} & \text{if $1 < d_{\mu} \leq n+1.$} 
			\end{cases}
	\]		
\end{lemma} 
For each $1 \leq k \leq \dim(\h)^{2}$, let  $\ket{s_{k}}$ be the $k$-th column of $S$; 
that is, $\ket{s_{k}} = \sum_{j} S_{j,k} \ket{j}$.
As $S$ is invertible, $\ket{s_{k}}$'s constitute a basis of the Hilbert space $\h^{\otimes 2} = \h \otimes \h$, and thus the vectorization of any quantum  state $\rho \in \dh$ can be uniquely represented as a linear combination of them: 
$ \ket{\rho} = \sum_{k} a_{k} \ket{s_{k}}$. 
Let
\[
	\maxmag{\rho} = \setcond{\lambda \in \maxmag{\e}}{\text{$M_{\e} \ket{s_{k}} = \lambda \ket{s_{k}}$ for some $k$ with $a_{k} \neq 0$}}
\]
be the set of eigenvalues of $M_{\e}$ with magnitude $1$ which contributes non-trivially to $\ket{\rho}$.  
Then we have the following lemma. 
\begin{lemma}
\label{lem:periodicstable}
	A QMC $(\h, \e, \rho_{0})$ is periodically stable if and only if $\maxmag{\rho_{0}}$ does not contain any element of the form $e^{i 2 \pi \psi}$ for some irrational number $\psi$.  
\end{lemma}
\begin{proof}
 First, note that for any $m > 0$ and $\rho \in \dh$, $\ket{\e^{m}(\rho)} = M^{m}_{\e} \ket{\rho}$. 
  Thus $(\h, \e, \rho_{0})$ is periodically stable if and only if there exists an integer $\theta > 0$ such that $\lim_{n \to \infty} M_{\e}^{n \theta} \ket{\rho_{0}}$ exists. 
  Let $\ket{\rho_{0}} = \sum_{j} a_{j} \ket{s_{j}}$. 
  For any $1 \leq j \leq \dim(\h)^{2}$, if $\ket{s_j}$ is an (generalized) eigenvector of $M_{\e}$ corresponding to an eigenvalue with magnitude strictly smaller than $1$, then $\lim_{n \to \infty} M_{\e}^{n\theta} \ket{s_{j}} = 0$ for any $\theta$. 
  Thus we only need to care about $\ket{s_j}$ corresponding to eigenvalues with magnitude one. 
  Following~\cite[Lemma 2]{guan2017super}, $\p \circ \e$ shares with $\e$ the same eigenvalues with magnitude $1$ and the corresponding eigenvectors, where $\p(\functiondot) = P \cdot P$, and $P$ is the projector onto the support of the maximal stationary state, no other stationary state supported in it. 
  Therefore, w.l.o.g, we assume $\e$ has a full-rank stationary state. 
  This kind of $\e$ is called faithful in~\cite{Albert2018}.

  Furthermore, for faithful $\e$, the Kraus operators $\setnocond{E_{i}}$ admit a diagonal form with respect to an appropriate decomposition $\h = \oplus_{k=1}^{m} \h_{k}$~\cite{ying2013reachability}. 
  To be specific, for all $i$,
  \[
    E_{i} = \oplus_{k=1}^{m} E_{i,k} = 
      \left[
        \begin{array}{cccc}
          E_{i,1} & & &\\
          & E_{i,2} & & \\
          & & \ddots & \\
          & & & E_{i,m} 
        \end{array}
      \right]
  \]
  where $E_{i,k}\in \l(\h_k)$,
  so $M_{\e}$ has the corresponding structure
  \[
    M_{\e} = \bigoplus_{k,l} M_{k,l}
  \]
  where $M_{k,l} = \sum_{i} E_{i,k} \otimes E_{i,l}^{*}$. Furthermore, for any $\theta >0$ we have
  \begin{equation}
    \lim_{n \to \infty} M_{\e}^{n\theta} \ket{\rho_{0}} = \bigoplus_{k,l} \lim_{n \to \infty} M_{k,l}^{n\theta} \ket{\rho_{0,k,l}},
    \label{eq:tmp}
  \end{equation}
  where $\ket{\rho_{0,k,l}}$ is the restriction of $\ket{\rho_{0}}$ onto $\h_{k} \otimes \h_{l}$, i.e., $\ket{\rho_{0}} = \oplus_{k,l} \ket{\rho_{0,k,l}}$.
  Now it is to see that $\lim_{n \to \infty} M_{\e}^{n\theta} \ket{\rho_{0}}$ exists if and only if for any $k$ and $l$, $\lim_{n \to \infty} M_{k,l}^{n\theta} \ket{\rho_{0,k,l}}$ exists.

  To verify the existence of $\lim_{n \to \infty} M_{k,l}^{n\theta} \ket{\rho_{0,k,l}}$, let 
  \[
    \ket{\rho_{0}} = \sum_{j} a_{j} \ket{s_{j}} = \bigoplus_{k,l} \sum_{j} a_{j} \ket{s_{j,k,l}},
  \]
  where $\ket{s_{j,k,l}}$ is the restriction of $\ket{s_{k}} $ onto $\h_{k} \otimes \h_{l}$. 
  Then $\ket{\rho_{0,k,l}} = \sum_{j} a_{j} \ket{s_{j,k,l}}$. Define 
  \[
    \maxmag{\rho_{0,k,l}} = \setcond{\lambda \in \maxmag{M_{k,l}}}{\text{$M_{k,l} \ket{s_{j,k,l}} = \lambda \ket{s_{j,k,l}}$ for some $j$ with $a_{j} \neq 0$}}.
  \]
  For any $1 \leq j \leq \dim(\h)^{2}$ and $1 \leq k,l \leq m$, we have two cases to consider:
  \begin{itemize}
  \item 
    If $\ket{s_{j,k,l}}$ is a (generalized) eigenvector of $M_{k,l}$ corresponding to an eigenvalue with magnitude strictly smaller than $1$, then $\lim_{n \to \infty} M_{\e}^{n\theta_{j}} \ket{s_{j,k,l}} = 0$ for any $\theta_{j}$. 
  \item 
    If $\ket{s_{j,k,l}}$ is an eigenvector of $M_{i,j}$ corresponding to an eigenvalue with magnitude $1$, then $M_{k,l}^{n} \ket{s_{j,k,l}} = e^{i 2 \pi \psi_{j,k,l} n} \ket{s_{j,k,l}}$ for some $\psi_{j,k,l}$. 
    Thus we have $\lim_{n \to \infty} M_{\e}^{n\theta_{j}}\ket{s_{j,k,l}}$ exists for some $\theta_{j} > 0$ if and only if $\psi_{j,k,l}$ is rational and $\theta_{j} \psi_{j,k,l}$ is an integer~\cite{hardy1979introduction}.
  \end{itemize}

  Note that the matrices $M_{k,l}$'s have the following spectral properties (cf.~\cite{guan2018structure}):
  for any $k$ and $l$, $\maxmag{M_{k,l}} = \emptyset$ or $\maxmag{M_{k,l}} = \setnocond{\exp(i 2 \pi (r + \psi_{k,l})/N_{k,l})}_{r=0}^{N_{k,l}-1}$, where $N_{k,l}$ is a positive integer and $\psi$ is a real number. 
  Thus $\lim_{n \to \infty} M_{k,l}^{n\theta} \ket{\rho_{0,k,l}}$ exists if and only if $\maxmag{\rho_{0,k,l}}$ does not contain any element of the form $e^{i 2 \pi \psi}$ for some irrational number $\psi$. 
  We complete the proof by noting that $\maxmag{\rho} = \setcond{\lambda \in \maxmag{\rho_{0,k,l}}}{1 \leq k,l \leq m}$.
  \qed
\end{proof}
%


The proof of Lemma~\ref{lem:periodicstable} gives us a way to compute the period of a periodically stable QMC.
\begin{corollary}
\label{cor:pstable}
	Let $\g = (\h, \e, \rho_{0})$ be a periodically stable QMC. 
	\begin{itemize}
	\item 
		If we rewrite $\maxmag{\rho_{0}}$ in the form
		\[
			\maxmag{\rho_{0}} = \setcond{e^{2 \pi i p_{k}/q_{k}}}{\text{$p_{k}$ and $q_{k}$ are coprime positive integers}}_{k},
		\]
		then $p(\g) = \lcm\setnocond{q_{k}}_{k}$.
	\item	
		For any integer $\theta > 0$, $\lim_{n \to \infty} \e^{n\theta}(\rho_{0})$ exists if and only if $p(\g)$ is a divisor of $\theta$.
	\end{itemize}
\end{corollary}


In the end, we present the proof of Lemma~\ref{mainlemmaQMC}

\begin{proof} 
  Let $\theta = p(\g)$.
  By the proof of Lemma~\ref{lem:periodicstable}, $\lim_{n \to \infty} \e^{n\theta+k}(\rho_{0})=\e_{\phi}(\e^{k}(\rho_{0}))$ for any integer $0\leq k\leq \theta-1$, i.e.
  $$\lim_{n \to \infty} \e^{n}(\rho_{0})=\e_{\phi}(\e^{n \modulo \theta }(\rho_{0})).$$ 
  Thus for any $\epsilon>0$, there exists a positive integer $K^\epsilon$ such that for all $n>K^\epsilon$, 
  \[\norm{ \e^{n}(\rho_{0})-\e_{\phi}(\e^{n \modulo \theta }(\rho_{0}))}<\epsilon\]
  as desired.

  Corollary~\ref{cor:pstable} provides a method to obtain $p(\g)$ by computing the Jordan decomposition of $\e$, of which the time complexity is $\bigO{n^{8}}$. 
  To determine $K^{\epsilon}$, we recall from Lemma~\ref{lem:inequality} that
  \[
    C^{-1} \mu^{n} n^{d_{\mu}-1} \leq \norm{M_{\e}^{n} - M_{\e_{\psi}}^{n}} \leq C \mu^{n} n^{d_{\mu}-1}.
  \]
  Let $n=m\theta+k$ with $0\leq k\leq\theta-1$, and  note that $\e_{\psi}^{m \theta}(\rho_{0}) = \e_{\phi}(\rho_{0})$ by Corollary~\ref{cor:pstable}. We have
  \[\norm{\e^{m \theta+k}(\rho_{0}) - \e_{\phi}(\e^k(\rho_{0}))} = \norm{\e^{m \theta+k}(\rho_{0}) - \e_{\psi}^{m \theta+k}(\rho_{0})} \leq \norm{\e^{m\theta} - \e_{\psi}^{m\theta}}\]
  where the inequality follows from  Eqs.~\eqref{eq:contractive} and~\eqref{eq:operatornorm}. Let $K^\epsilon=M^\epsilon\theta.$
  So we can simply set $M^{\epsilon}$ to be the minimal integer satisfying 
  \begin{eqnarray}
    \text{$C \mu^{M^{\epsilon} \theta} {(M^{\epsilon} \theta)}^{d_{\mu}-1} < \epsilon$ \quad and \quad $M^{\epsilon} \theta + 1 > d_{\mu}$},
    \label{Eq_Me}
  \end{eqnarray}
  where the second inequality comes from the requirement of $C$ in Lemma~\ref{lem:inequality}. Finally, the computation of $C$ boils down to the Jordan decomposition of $\e$, which make the time complexity of calculating $K^{\epsilon} = M^{\epsilon} \theta$ be $\bigO{n^{8}}$.
  \qed
\end{proof}
\section{Reduction from Problem~\ref{prob:quantum} to Problem~\ref{prob:approx}}\label{Sec:reduction}
Let 
\[
    \g = (\h,\e)
    \qquad
    \g' = (\h^{\otimes 2}, \e \otimes \mathrm{id}_{\h}, \ket{\Omega} \bra{\Omega}/d)
\]
where $d = \dim(\h)$ and $\ket{\Omega} = \sum_{k=1}^{d} \ket{k}\ket{k}$, where $\setnocond{\ket{k}}$ is an orthonormal basis of $\h$. 
Let $\lfst \colon \dh \to 2^{\ap}$ be the labeling function such that for any $\f \in \sh$,
\[
    \lfst((\f \otimes \mathrm{id}_{\h})(\ket{\Omega} \bra{\Omega}/d)) = \lfdy(\f).
\]
Then it is easy to check that $\lfst(\trjst(\g')) = \lfdy(\trjdy(\g))$. 
Furthermore, as the following function is bijective:
\begin{eqnarray*}
    \sh \to D(\h^{\otimes 2}) \colon \f \mapsto (\f \otimes \mathrm{id}_\h)(\ket{\Omega}\bra{\Omega})/d,
\end{eqnarray*}
$\epsneigh(\trjst(\g')) = \epsneigh(\trjdy(\g))$ for all $\epsilon$ if we choose a norm on $\sh$ as $\norm{\e - \f}_{\diamond} = \norm{(\e \otimes \mathrm{id}_{\h})(\ket{\Omega}\bra{\Omega})/d - (\f \otimes \mathrm{id}_{\h})(\ket{\Omega}\bra{\Omega})/d}$ and $\epsneigh(\f) = \setcond{\lfdy(\f') \in 2^{\ap}}{\f' \in \sh, \norm{\f - \f'}_{\diamond} < \epsilon}$.
\section{Solving Problem~\ref{prob:quantum}}\label{Section:solvemanybody}

Similar to Lemma~\ref{lem:periodicstable} and Corollary~\ref{cor:pstable}, we can show that $\e$ is periodically stable if and only if $\e$ does not have any eigenvalue of the form $e^{i 2 \pi \psi}$ for some irrational number $\psi$.
Furthermore, if $\e$ is periodically stable and $\maxmag{\e} = \setcond{e^{2 \pi i p_{k}/q_{k}}}{\text{$p_{k}$ and $q_{k}$ are coprime positive integers}}_{k}$, then $p(\e) = \lcm\setnocond{q_{k}}_{k}$. 
Finally, for any integer $\theta$, $\lim_{n \to \infty} \e^{n \theta}$ exists if and only if $p(\e)$ is a divisor of $\theta$.

Recall that a super-operator $\e \in \sh$ is called \emph{irreducible} if it has only one full-rank stationary state. 
From~\cite[Theorem 6.6]{wolf2012quantum}, the peripheral (magnitude $1$) spectrum of such an $\e$ has a nice structure: 
$\maxmag{\e} = \setcond{e^{2 \pi i k/m}}{1 \leq k \leq m}$ for some integer $m \leq \dim(\h)^{2}$. 
Thus $\e$ is periodically stable and $p(\e) = m$.

The following lemma, which is analogous to Lemma~\ref{mainlemmaQMC} for QMCs, is crucial.

\begin{lemma}
\label{lem:mainlemmaSupOp}
    Given a periodically stable super-operator $\e$ with period $p(\e)$, let $\e_{k} = \e_{\phi} \circ \e^{k}$, for $0 \leq k < p(\e)$.
    Then for any $\epsilon > 0$, there exists an integer $K^{\epsilon} > 0$ such that for any $n \geq K^{\epsilon}$,
    \[
        \lfdy(\e^{n}) \in \epsneigh(\e_{n \modulo p(\e)})
    \]
    where $\epsneigh(\f) = \setcond{L(\f') \in 2^{\ap}}{\f' \in \sh, \norm{\f - \f'} < \epsilon}$. Furthermore, the time complexities of computing $p(\e)$ and $K^{\epsilon}$ are both in $\bigO{d^{8}}$, where $d = \dim(\h)$.
\end{lemma}

With the above Lemma, $K^\epsilon$ can always be set as a multiple of $p(\e)$. Then we get an $\omega$-expression 
\[
    \epsneigh(\trjdy(\g)) = \setnocond{\lfdy(\mathrm{id}_\h)} \cdot \setnocond{\lfdy(\e)} \cdots \setnocond{\lfdy(\e^{K^{\epsilon}-1})} \cdot \left(\epsneigh(\e_{0}) \cdots \epsneigh(\e_{p(\e)-1})\right)^{\omega}.
\]


Furthermore, the approximate version of Problem~\ref{prob:quantum} can be solved.

\end{document}